\documentclass[11pt]{article}

\usepackage{amsmath,amsfonts,amsthm,amssymb,graphicx,multicol,amscd}

\theoremstyle{definition}
\newtheorem{definition}{Definition}
\newtheorem{fallacy}{Fallacy}
\theoremstyle{remark}
\newtheorem{remark}[definition]{Remark}
\newtheorem{example}[definition]{Example}

\newtheoremstyle{mytheorem}{0.5cm}{0.2cm}{\slshape}{ }{\bfseries}{.}{ }{}
\theoremstyle{mytheorem}
\newtheorem{theorem}{Theorem}

\newtheorem{lemma}{Lemma}

\renewcommand{\phi}{\varphi}

\newcommand{\edge}[1]{\ar@{-}[#1]}
\newcommand{\lulab}[1]{\ar@{}[l]_<<{#1}}
\newcommand{\rulab}[1]{\ar@{}[r]^<<{#1}}
\newcommand{\ldlab}[1]{\ar@{}[l]^<<{#1}}
\newcommand{\rdlab}[1]{\ar@{}[r]_<<{#1}}

\begin{document}

\title{Scenarios and their Aggregation in the Regulatory Risk Measurement Environment}

\author{Andreas Haier %\thanks{FINMA, Bern, Switzerland ({\tt andreas.haier@finma.ch})}\ \,
    and Thorsten Pfeiffer\thanks{Both FINMA, Bern, Switzerland.
    The authors wish to point out that they express solely their personal
    beliefs. \newline This manuscript was submitted to an Actuarial Journal on August 16, 2012.\newline
    ({\tt andreas.haier@finma.ch, thorsten.pfeiffer@finma.ch})}}
    
\date{August 16, 2012}

\maketitle

\begin{abstract}
We define scenarios, propose different methods of aggregating them,
discuss their properties and benchmark them against quadrant requirements.

  \medskip


  %\noindent{AMS Classifications}: 60E05; ????
\end{abstract}

\section{Terminology}

We assume that the \emph{available capital} of an insurance company
can be described as a function of observable real variables, called
\emph{risk factors}.

Examples of risk factors are basic economic variables such as interest
rate (yield curve), share index, real estate, FX (foreign exchange) and
corporate spreads.

The dependency of the available capital on these risk factors is
described by the \emph{valuation function} $V$  of an insurance company: %of a financial object
\[V: \mathbb{R}^{n} \rightarrow
\mathbb{R}\] where \[ (x_1, \dots , x_n) \mapsto V(x_1, \dots ,
x_n).
\]

We require $V$ to be \emph{measurable} with respect to the Borel
$\sigma$ algebras $(\mathfrak{B}_{n}, \mathfrak{B})$, and we
consider $\mathbb{R}^{n}$ together with a probability measure $P$
with respect to the Borel $\sigma$ algebra $\mathfrak{B}_{n}$. This measure can be
thought of as representing the distribution of the risk factors at a
given future point in time. Define $P_{V} = V_* P$ as image measure
on $\mathbb{R}$; this describes the distribution of the risk bearing
capital. Sometimes we us the notion \emph{risk bearing capital} as
synonymous to the one of available capital.

\section{Risk measurement vs risk measurement for regulatory purposes} \label{Riskmeasure}
Risk measurement requires the following steps to be performed:
\begin{enumerate}
\item Identifying a suitable set of risk factors
\item Making an assumption about the distribution of risk factors (i.e.\ choosing $P$)
\item Describing the dependency of available capital on the risk factors (i.e.\ defining $V$)
\item Analysing $P_{V}$
\end{enumerate}

Steps $3.$ and $4.$ are clearly entity specific.
In a regulatory environment, steps $1.$ and $2.$ may be subject to minimum requirements
from the supervisor to guard against lack of awareness or misjudgements by the company's management.

Regarding step $1.$, there are basic economic variables like
interest rates, which may have a major impact on most entities.
Therefore supervisors require them to be considered. Additional risk
factors, which are important only to a specific entity, may be
included in the analysis by the affected entities (e.g.\  exotic
asset classes).

Regarding step $2.$, for purposes of equal treatment of all
supervised entities, it is important to note that the distribution
of the ``basic'' risk factors prescribed by the supervisor for
regulatory purposes is identical for all supervised entities in one
jurisdiction, although there may be different views on these
distributions.

For example, interest rates in one year's time will be the same for
all entities, although they may affect different entities to a
different degree. So, the distribution of the risk factor interest
rate should be the same for all entities; what is different, though,
is how the individual $V$ acts on that distribution (step $3).$
In short, whilst $P$ should be driven by (minimum)
requirements from the supervisor, $P_{V}$ is individually determined
based on $V.$

In a framework for determining \emph{economic capital}, which
is a fully internal exercise, steps $1.$ to $4.$ are clearly \emph{not}
subject to regulatory requirements. Sometimes, this leads to a
misunderstanding, which we formulate in
\begin{fallacy}\label{Fallacy1}
We (company management) have our own view on the interest rates in one year's time which
we use as well for our risk management and economic capital purposes, and we want to use
it for the regulatory capital requirements as well.
\end{fallacy}

Following such an approach in general would lead to regulators being unable to prescribe
the estimation of the future behaviour of risk factors. The fallacy arises because,
by definition, regulators set regulatory requirements.

\section{Mathematically elegant way of setting requirements on risk factors} \label{Mathematicallyelegant}

Usually, supervisors are interested in the behaviour of entities
under extreme events. In this section, we introduce a very simple
toolbox which can be used to construct and describe extreme events.

\begin{definition}\label{quadrants}
For each non-zero linear form $\lambda : \mathbb{R}^{n} \rightarrow
\mathbb{R}$ and $c \in \mathbb{R}$, a set of the form
$\lambda^{-1}([c,\infty))$ is called an \emph{affine half-space}. A
\emph{quadrant} is a non-empty intersection of a finite set of
affine half-spaces.
\end{definition}

By definition, every quadrant is an element of $\mathfrak{B}_{n}$,
and $\mathfrak{B}_{n}$ is generated by the set of all quadrants.

\begin{example}
Any affine subspace of $\mathbb{R}^n$ is a quadrant. In particular,
a subset consisting of one point only is a quadrant.
\end{example}

\begin{definition}
Let $A_Q$ be a quadrant and $p_Q \in [0,1]$. A \emph{quadrant
requirement} is a couple $Q=(A_Q,p_Q)$. $P$ is said to \emph{fulfill
the quadrant requirement $Q$} if it satisfies the following
condition
\[
P(A_Q) \geqq p_{Q}
\]
\end{definition}

\begin{remark}
Even if it would be possible to formulate requirements on more general sets than quadrants,
we feel that the setting chosen above will be sufficient for regulatory purposes.
\end{remark}

In setting quadrant requirements, supervisors can express their
judgement on the future behaviour of risk factors. They can ensure that
the measures $P$ used by the entities have sufficient weight in the
tail of the common distribution of the risk factors. Suppose, for example,
a supervisor believes that the following statement is relevant for
risk measurement for regulatory purposes:

\begin{quote}
``In one year's time, the interest rate of the ten year's Swiss
frank is less or equal to 0.5\% with a probability of 1\%.''
\end{quote}

He could translate this statement to the following quadrant
requirement:
 \begin{quote}
 ``$P(i_{10} \leqq .5\%) \geqq 1\%$'',
 \end{quote}
  where the variable $i_{10}$ describes the return rate of the ten year's Swiss frank. This quadrant requirement should be the same and equal for all entities under supervision, what is different is the impact of that requirement on the individual insurer's available capital. Suppose a Company~$1$ which is completely hedged against movements of the ten year Swiss rate, thus $\Delta V_{1}(i_{10})=0$, where $\Delta V$ describes the change of the available capital. A Company~$2$, which refrains from hedging, is likely to have the result $\Delta V_{2}(i_{10})<0$.

\begin{definition}

A set of quadrant requirements $M$ is an arbitrary finite set of
quadrant requirements such that the number $p_M=\sum_{Q\in
M}{p_Q}\leqq 1$.

\end{definition}

 A set of quadrant requirements is useful for a supervisor to set his requirements on more than one risk factor, and helpful to formalize his judgement on tail dependencies. This does not mean that all companies in one jurisdiction have to use the same $P$, it only means that the set of acceptable $P$ is restricted by the supervisor.

 Quadrant requirements have also the appealing property that it is easy to check in an objective and reproducible manner whether or not they are fulfilled. Indeed, subjective judgement is reduced to such an extent that an independent third party or law court could easily double check the supervisor's assessment and come to the same result.

 There might be as well other regulatory requirements on $P$ such as e.g.\  being ``realistic'', ``state of the art'', and other ``qualitative'' criteria. These are very important to gain a better mutual understanding of $P$ and should be intensively discussed in the regulatory dialogue between the company and the supervisor. The constant challenge with these ``qualitative'' criteria is that, due to their inherent element of subjectivity, an independent third party or law court does not necessarily need to come to the same result as the supervisor. A way out could be transforming the rather subjective ``qualitative'' criteria into objectively testable quadrant requirements.

 We conclude the section with quantifying the statement that ``a supervisor should neither be substantially under- nor over-prescriptive'' in terms of the number of quadrant requirements to be used. Suppose a supervisor wants one and only one $P$ to be used in his jurisdiction. How many quadrant requirements would he need? An answer is given by:

\begin{theorem}
Let $P$ be any probability measure on
$(\mathbb{R}^n,\mathfrak{B}_n)$. Then there exists a countable set
of pairs $(A_{i},p_{i})$ of quadrants $A_{i}$ and real numbers
$p_{i} \in [0,1]$ such that $P$ and only $P$ fulfills all quadrant
requirements $(A_{i},p_{i})$.
\end{theorem}

\begin{proof}
Consider hypercubes whose corners are in $\mathbb{Q}^n$. These
hypercubes are quadrants, and the set of these hypercubes is
countable. Let $A_i$ be an enumeration of these hypercubes. Define
$p_i$ by $p_i=P(A_i)$. Then by definition of $p_i$, $P$ satisfies
all quadrant requirements $(A_i,p_i)$.

Conversely, assume $P'$ satisfies all quadrant requirements
$(A_i,p_i)$. Consider the set $\mathfrak{M}$ of all $X \in
\mathfrak{B}_{n}$ for which $P(X) = P'(X)$. Then $\mathfrak{M}$ is a
$\sigma$ algebra. Also, $A_i \in \mathfrak{M}$. Indeed, $P'(A_i)
\geq P(A_i)$ by definition of the quadrant requirements. On the
other hand, $\mathbb{R}^n \backslash A_i$  can be written as a union
of at most countably many disjoint hypercubes $A_j$ where $j\in
J(i)$, hence
\[
\begin{aligned}
1- P'(A_i) = P'(\mathbb{R}^n \backslash A_i) = \sum_{j\in J(i)} P'(A_j) \\
\geq \sum_{j\in J(i)} P(A_j) = P(\mathbb{R}^n \backslash A_i) = 1 -
P(A_i)
\end{aligned}
\]
This shows that $P'(A_i) \leq P(A_i)$, so we have $A_i \in
\mathfrak{M}$.

As $\mathfrak{B}_n$ is generated by $A_i$, we have
$\mathfrak{M}=\mathfrak{B}_n$. Thus, $P = P'$.
\end{proof}

%A little weaker is following
%
% \begin{remark}
% Let $P$ be any probability measure on $(\mathbb{R}^n,\mathfrak{B}_n)$ and $\epsilon > 0$. Then there are countably many pairs $(A_{i},p_{i})$ of quadrants $A_{i}$ and real numbers $p_{i} \in [0,1]$ such that
% \begin{itemize}
% \item $\mathbb{R}^n = \dot \cup_{i \in I} A_{i}$,
% \item $\sum_{i} p_{i} = 1$ and
% \item $\forall B \in \mathfrak{B}_n$: $\vert P(B) - \bar P'(B) \vert < \epsilon$
% \end{itemize}
% for a probability measure $\bar P'$, which is an extension on $(\mathbb{R}^n,\mathfrak{B}_n)$ of $P'(A_{i}):=p_{i}$.
% \end{remark}

 Now we can classify as follows:
 \begin{itemize}
\item \emph{No quadrant requirements at all} is equivalent to \emph{The distribution $P$ is exclusively the choice of each supervised entity} (no regulatory prescription at all, cf. Fallacy~\ref{Fallacy1}).
\item \emph{Countably many requirements on quadrants} can be defined so that \emph{the distribution $P$ is determined exclusively by the supervisor} (supervisor could well be regarded as ``overly prescriptive'').
\end{itemize}

\section{Scenarios and case study: use of scenarios in the SST}\label{casestudy}

The use of scenarios and stress testing is currently much discussed
within the regulatory community. We give a general definition of
scenarios and briefly describe their use under the Swiss Solvency
Test (SST).

\begin{definition}

A scenario $s$ is an element $d_s \in \mathbb{R}^n$, and an impact
of the scenario $s$ is the value $V(d_s)$; $d_s$ is also called
\emph{stressed situation}.

\end{definition}

$d_s$ can be thought of as a concrete realization of the underlying
risk factors. It is a powerful tool in the ``what-if-analysis'', to
answer the question:

\begin{quote}
"What happens to Company $1$ under a certain change of risk
factors?"
\end{quote}

In order to do so, the supervisor might specify $d_s$, the ``if'',
the stressed situation, and Company $1$ evaluates the impact of the
scenario $d_s$, gives the ``what'' by calculating $V_{1}(d_s)$. This
is ``qualitative'' in the sense that the impact of scenarios is
discussed between the entity and the supervisor. The impact of a
scenario might or might not increase the regulatory prescribed
capital required (PCR). Switzerland is a jurisdiction where scenarios impact the
PCR, and we briefly describe how this is done under
the current Swiss regime, the SST:

\begin{definition}

An \emph{enhanced scenario} $S$ is a couple $S=(d_S,p_S)$ consisting
of $d_S \in \mathbb{R}^n$ and $p_S \in [0,1]$, where $p_S$ is called
probability of occurence of $S$, $d_S$ is called \emph{deflection}
of $S$.

\end{definition}

\begin{definition}

A scenario set $M$ is an arbitrary finite set of enhanced scenarios
such that $p_M=\sum_{S\in M}{p_S}\leqq 1$.

\end{definition}

The scenario is enhanced in the sense that it comes with a
probability of occurence, which is set by Swiss supervisors. It is
needed to calculate the \emph{target capital}, as the PCR is
referred to in the SST, and usually more than one scenario is used.
In Switzerland, the target capital is determined via a distribution-based
approach.
The way of taking scenarios into account in the target capital
is to modify the initially assumed distribution in a certain way,
which is usually referred to as \emph{aggregation}.

In the following way Swiss supervisors currently aggregate the
impact of the scenarios. Strictly speaking, this method is only
prescribed within the SST standard model, but roughly speaking, it
currently constitutes as well a sort of industry standard in the
Swiss insurance sector, even for some entities using internal models
for regulatory purposes.

\begin{definition}[SST scenario aggregation]
Let $M$ be a scenario set, $P$ a probability measure on
$(\mathbb{R}^n,\mathfrak{B}_n)$ and $V$ a valuation function.
Consider the expression
\[(1-p_M) V_{*} P + \sum_{S\in M}{p_S ({\tau_{V(d_S)}}_{*} V_{*} P)},\]
where $\tau_{V(d)}$ denotes the translation in $\mathbb{R}$ in by
$V(d)$. This expression is called \emph{SST scenario aggregation}.
\end{definition}

Thus, the SST scenario aggregation is a \emph{mixture} of the
starting distribution $V_* P$ (describing the distribution of risk bearing capital \emph{before} aggregation of scenarios) with its translated versions, where
the translation is given by the impact of each scenario. However,
the fact that the valuation function $V$ gets involved during
scenario aggregation obscures the view on what is happening. As the
scenarios themselves are defined at the risk factor level, it seems
desirable to have a definition of what scenario aggregation means at
this level. This might be achieved by the following alternative approach:

\begin{definition}[aggregation by shifting]
Let $M$ be a scenario set, $P$ a probability measure on
$(\mathbb{R}^n,\mathfrak{B}_n)$. Define
\[P_M = (1-p_M) P + \sum_{S\in M}{p_S {\tau_{d_S}}_{*} P}\]
where $\tau_d$ denotes $v \mapsto v + d$, the \emph{translation} by
$d$ in $\mathbb{R}^{n}$. We say that $P_M$ is obtained from $P$ by
aggregating enhanced scenarios from scenario set $M$ to $P$.
\end{definition}

Thus, similar to the situation above, the resulting distribution of
the risk factors is a \emph{mixture} of the starting distribution
$P$ with its translated versions, where the translation is given by
the deflection of each scenario. So one might ask whether the SST
scenario aggregation at the level of the risk bearing capital leads
to the same result as the aggregation by shifting on the risk factor
level defined above. The following lemma deals with this question.

\begin{lemma}
The aggregation of enhanced scenarios acts on the distribution of
risk bearing capital in a similar way as on the distribution of risk
factors: the mixture of the distribution of the underlying risk
factors carries over to the distribution of risk bearing capital,
where the translation by the deflection of the scenarios is being
replaced by a translation by the impact of the scenarios, and the
valuation function $V$ is replaced by its twisted versions $V_d$.

The twisted valuation functions are defined by

\[
V_d(x):=V(x+d)-V(d)
\]

If $V$ is additive, then $V_d$ coincides with its twisted versions,
and in this case SST scenario aggregation leads to the same result
as the aggregation by shifting on the risk factor level.
\end{lemma}

\begin{proof}

Note that
%\[
%\begin{aligned}
\begin{align*}
V_{*} P_M &= V_{*} ((1-p_M) P + \sum_{S\in M}{p_S {\tau_{d_S}}_{*} P}) \\
&= (1-p_M) V_{*} P + \sum_{S\in M}{p_S V_{*} ({\tau_{d_S}}_{*} P)} \\
&= (1-p_M) V_{*} P + \sum_{S\in M}{p_S ({\tau_{V(d_S)}}_{*}
{V_{d_S}}_{*} P)}\,.
%\end{aligned}
%\]
\end{align*}

In the last step, we have used the fact that

\[
V \circ \tau_d = \tau_{V(d)} \circ V_d\,.
\]

If $V$ is additive, then clearly $V_d=V$ and the expression obtained
above coincides with the definition of SST scenario aggregation.
%$\diamond$
\end{proof}

So, for additive $V$ the aggregation by shifting carries over to the SST scenario aggregation. However, in risk measurement $V$ is generally not additive.\\

One further challenge with aggregation by shifting is, that the
company specific distribution $P$ serves as input. Companies
evaluate the difference between their starting $P$ and the $P_M$
after scenario aggregation to understand the impact of the
scenarios. If a starting $P_1$ from Company $1$ in terms of PCR is
``harder'' than $P_2$ of Company $2$, this carries over to the
distributions ${P_1}_M$ and ${P_2}_M$, if both companies are asked
by the supervisor to aggregate the same scenario set $M$. So it could
be regarded to be fairer, to allow for different scenarios sets
$M_1$ and $M_2$, which reflects in some way the differences between
the starting distributions $P_1$ and $P_2$. In such an environment,
it is very difficult for the supervisor to evidence equal treatment.
Furthermore, with many companies in one jurisdiction, it is not very
practicable for the supervisors to prescribe a different scenario set
for all the companies.

%On the other hand, companies might have the possibility to discard a scenario, provided they are able to convince the supervisor it is already "appropriately" reflected in the starting distribution $P$. This "scenario selection" might not always be undisputed between supervisor and supervised company, and a generally understandable or accepted methodology has not yet emerged.

In the following sections, it is our program to develop criteria
which overcome these challenges using the quadrant language from
section \ref{Mathematicallyelegant}.

\section{Translating quadrant requirements into the scenario language}\label{Translatingquadrant}

In this section, our aim is to show that if a supervisor would like
to impose quadrant requirements, he can achieve the same aim by
requiring the aggregation of scenarios.

Unfortunately, if we use
aggregation by shifting, this statement only holds in a weaker form and under additional
assumptions on the nature of the quadrant requirements. The deeper reason behind this is that the
mechanics of aggregation by shifting strongly depends on the
initially assumed distribution $P$.

Therefore, we initially introduce a more convenient, alternative
aggregation method, which we call ``point mass aggregation''. This
helps us establishing our statement in a more general context.

Only then we return to the SST scenario aggregation and investigate,
which additional assumptions are necessary to establish our
statement when aggregation by shifting is applied by the supervisor.

Basically, this means that a supervisor relying on aggregation by
shifting acts self-restrictive in terms of the quadrant requirements
he is able to impose.

\begin{definition}[point mass aggregation]\label{ptmassdef}
Let $M$ be a scenario set, $P$ a probability measure on
$(\mathbb{R}^n,\mathfrak{B}_n)$. Define
\[P_M^{pt} = (1-p_M) P + \sum_{S\in M}{p_S \delta_{d_S}}\]
where $\delta_d$ denotes the Dirac measure centered at $d$. We say
that $P_M^{pt}$ is obtained from $P$ by aggregating enhanced
scenarios from scenario set $M$ as point-mass to $P$.
\end{definition}

The name of this aggregation method should be intuitive: It means,
that scenarios are aggregated by adding point masses at each
scenario deflection with the corresponding probability; on the other
hand, the probability weight of the initial distribution $P$ is
reduced accordingly, so that the resulting measure is a probability
measure.

\begin{theorem} \label{maintheorem1}
\label{pt_mass_agg_thm} Let $N=\{Q_1,Q_2,\ldots,Q_k\}$ be a set of
quadrant requirements. Then there exists a scenario set $M$ such
that for any distribution $P$, the distribution $P_{M}^{pt}$
resulting from aggregating the enhanced scenarios from $M$ as
point-mass to $P$ satisfies the quadrant requirements from $N$.
\end{theorem}

\begin{proof}
Denote the quadrants and probabilities of the quadrant requirements
by $A_j$ and $p_j$, respectively. As $A_j\neq \emptyset$, we can
choose $d_j\in A_j$. $S_j=(d_j, p_j)$ defines an enhanced scenario,
and we define $M=\{S_1,S_2,\ldots,S_k\}$. Then for each $j$,
$P_M^{pt}$ has a point-mass of weight $p_j$ at $d_j$ by definition
of point-mass aggregation. As $A_j\supseteq \{d_j\}$,
$P_M^{pt}(A_j)\geq P_M^{pt}(\{d_j\})\geq p_j$. This shows that
$P_M^{pt}$ satisfies the quadrant requirements from $N$.
\end{proof}

Note that the theorem does not hold if we replace point mass
aggregation by aggregation by shifting, as can be seen from the
following

\begin{example} \label{Counterexample}
Let $p_{\max} \in [0,1]$ and $P$ such that $P(B) < p_{\max}$ for all
balls $B$ with fixed radius $R \geq 0$. For any quadrant $A$
contained in such a ball $B$ with radius $R$, $P$ never fulfills the
quadrant requirement $(A, p_{\max})$. Additionally, for any scenario
set $M$, $P_M$ does not satisfy the quadrant requirement
$(A,p_{\max})$. Indeed, aggregation by shifting is based on
translations, and the radius of any ball is invariant under
translations. Thus for any ball $B$ we have

\[
\begin{aligned}
P_M (B)& = (1-p_M) P(B) + \sum_{S\in M}{p_S {\tau_{d_S}}_* P(B)} \\
&=(1-p_M) P(B) + \sum_{S\in M}{p_S P(\tau_{-d_S} (B))} \\
&< (1-p_M) p_{\max} + \sum_{S\in M}{p_S p_{\max}} = p_{\max}
\end{aligned}
\]

Hence $P_M$ also has the property that $P(B) < p_{\max}$ for all
balls $B$ with fixed radius $R \geq 0$, so the assertion follows.

\end{example}

The situation above is easily achieved, as we show in the following

\begin{example}
Assume that $P$ is a distribution with density with respect to the
Lebesgue measure.

% such that $P(\vert \vert x \vert \vert > S)\neq 0$ for any $S\in [0,\infty)$.

Assume further that we have a set of quadrant requirements, such
that all quadrants are bounded and contained in a ball of radius $R
\geq 0$. Define $p_{\max} = \max_j(p_j)>0$, where $p_j$ are the
probabilities associated to the quadrant requirements, and denote
the quadrant associated to the maximum by $A$.

Let $\mathrm{Mult}_c$ denote multiplication by $c\in\mathbb{R}$
within $\mathbb{R}^n$. Then, we first observe that there exists
$c\in \mathbb{R}$ such that the image measure
$P'={\mathrm{Mult}_c}_* P$ has the property that $P'(B) < p_{\max}$
for \emph{any} ball $B$ of radius $R$.

Indeed, we can find a suitably countable partition of $\mathbb{R}^n
= \dot \cup_{i \in I} A_{i}$ such that
\begin{itemize}
\item $P(A_{i}) <\frac{p_{\max}}{2^n}$ $\forall i$
\item all $A_{i}$ are contained in balls of radius $\vert r \vert > 0$
\item for any ball $B$ with radius $\vert r \vert > 0$ there is $J \subseteq I$ with $\vert J \vert \leq 2^n$ and $B \subseteq \dot \cup_{j \in J} A_{j}$
\end{itemize}

We can then choose $c=R/r$, and it is easily verified that this $c$
has the required property from example \ref{Counterexample}.

\end{example}

Astonishingly, even for an easy distribution family as the normal,
it is not always possible to find a scenario set if we use the
aggregation by shifting.

We thus found a contradiction to following

\begin{fallacy}
One can always correct $P$ by aggregating scenarios by shifting in a
way that prescribed quadrant requirements are satisfied.
\end{fallacy}

However, under stronger assumptions on the quadrant requirements, we
are able to establish a weaker result also for aggregation by
shifting.

\begin{definition}\label{twosidedconstrained}
A convex set is said to be two-sided constrained, if it is contained
in a suitable set of the form $\lambda^{-1}([a,b])$, where $0 \neq
\lambda$ is a linear form on $\mathbb{R}^n$ and $a,b\in\mathbb{R}$.
\end{definition}

A quadrant being two-sided constrained basically means, that it is
defined by two counteracting conditions. An example could be a
quadrant defined by the 10-year-rate being in the range between
$0.5\%$ and $1\%$. In the regulatory context, especially in cases
where it is known which direction of move of a risk factor is
adverse, it is more common not to consider quadrants which are
two-sided constrained, as one does not wish to exclude ``more
extreme'' situations from the consideration; however, it may still
be useful to use quadrants which are two-sided constrained in some
cases. As an example, by buying derivatives a company might change
its risk profile such that the valuation function $V$ is no longer
monotonic in a certain risk factor, but known to have a minimum in
an interval $[a,b]$. In this situation, it might be useful to
consider events from a quadrant which is two-sided constrained.

We briefly discuss what definition \ref{twosidedconstrained} means
for the real numbers in

\begin{example}\label{linearformsonR}
%The linear forms $\lambda_{l}: \mathbb{R} \rightarrow \mathbb{R}$ are exactly the mappings $v \mapsto l * v$ for $l \in \mathbb{R}$. If $l=0$, we obtain the zero linear form, $l=1$ gives the identity, and $l=-1$ the reflection of the real line in $0$.
%For $l > 0$ we obtain a scaling map and in case $l < 0$ a reflecting scaling map.
%For every $a,b \in \mathbb{R}$ we can use the identity map to see that both $[a,b]$ and the symmetric interval $[- \vert a \vert, \vert a \vert]$ are two-sided constrained.
Consider a connected $\emptyset \neq S \subseteq \mathbb{R}$. If $ - \infty < \mathrm{inf}(S)$ and $ \mathrm{sup}(S) < \infty$, then
$S \subset [ \mathrm{inf}(S) - 1, \mathrm{sup}(S) + 1]$ is two-sided constrained. If $\mathrm{inf}(S) = - \infty$ or $ \mathrm{sup}(S) = \infty$, then $S$ contains intervals of arbitrary length. Thus, a convex $S \subseteq \mathbb{R}$ is either two-sided constrained or it contains line segments of arbitrary length.
\end{example}

We can now come to

\begin{theorem}
\label{SST_aggr_thm} Let $Q$ be a set of quadrant requirements
satisfying the following two additional conditions
\begin{enumerate}
\item \label{cond1} None of the quadrants of $Q$ is two-sided constrained
\item \label{cond2} The cumulative probability satisfies $p_Q<1$
\end{enumerate}

Then for any distribution $P$, there exists a scenario set $M$ such
that the distribution $P_M$ satisfies the quadrant requirement from
$Q$.

\end{theorem}

\begin{remark}
Note that this result is much weaker than
Theorem~\ref{pt_mass_agg_thm}, as the scenario set chosen depends on
$P$.
\end{remark}

For the proof of Theorem~\ref{SST_aggr_thm} we will need some
properties of quadrants, which we formulate in a more general setting for convex sets as

\begin{lemma} \label{convex sets}

Let $S \subseteq \mathbb{R}^n$ be a convex set. Then exactly one of the
following two statements holds

\begin{enumerate}
\item $S$ is two-sided constrained.
\item $S$ contains $n$-dimensional balls of arbitrarily large radius as subsets.
\end{enumerate}
%Still need a proof or reference here: if not, leave away lemma and make it a condition.
\end{lemma}

The proof can be found in the appendix.

\begin{remark}
We have not specified the norm with respect to which we consider the
balls in the lemma. However, if $S$ contains arbitrarily large balls
with respect to one norm, it will contain arbitrarily large balls
with respect to any norm, as all norms on $\mathbb{R}^n$ are
equivalent.
\end{remark}

The idea of the proof of Theorem~\ref{SST_aggr_thm} can now be
briefly outlined as follows: Given $P$, by its $\sigma$-additivity,
we can find a large ball of radius $R$ such that ``most'' of the
mass of $P$ lies inside this ball. By Condition \ref{cond1} and
Lemma~\ref{convex sets} we can shift $P$ by a ``deflection'' such
that this ball lies inside a given quadrant. By taking this
deflection to define a scenario, we cannot achieve that ``all'' the
mass lies inside the corresponding quadrant, but most of it. By then
increasing the probabilities associated to the quadrants only
slightly such that the sum remains below $1$ - which can be achieved
due to Condition \ref{cond2} of the theorem - we can ensure that
sufficient mass lies inside each quadrant, and therefore the
quadrant requirement will be satisfied.

\begin{proof}
Denote the quadrant requirements from $Q$ by $(A_i, p_i)$, and let
$N$ be the number of quadrants. As $p_Q<1$, we can choose $\epsilon
> 0$ such that $\frac{p_Q}{1-\epsilon} < 1$. Because of $\sigma$
additivity and $\mathbb{R}^n = \cup_{R\in \mathbb{N}} B_R$, we have
$lim_{R \rightarrow \infty} P(B_R)=1$, where $B_R$ denotes a ball of
radius $R$ centered at $0$. Hence, we can choose $R$ such that
$P(\mathbb{R}^n \backslash B_R)< \epsilon$.

By Condition \ref{cond1} and Lemma~\ref{convex sets} we can find a
deflection $d_i$ for $i=1\ldots N$, such that the ball of radius $R$
centered at $d_i$ lies inside $A_i$.

Take $M$ to be the set of enhanced scenarios $(d_i, p'_i)$, where
$p'_i = \frac{p_i}{1-\epsilon}$. We claim that $P_M$ satisfies the
quadrant requirements from $Q$, as desired. Indeed, for any $A_i$ we
have

\[
\begin{aligned}
P_M(A_i) = (1-p_M) P(A_i) + \sum_{S\in M}{p'_S {\tau_{d_S}}_* P(A_i)} \\
=(1-p_M) P(A_i) + \sum_{S\in M}{p'_S P(\tau_{-d_S} (A_i))} \\
\geq p'_{i} P(\tau_{-d_{S_i}} (A_i)) \geq p'_{i} P(B_R)
> p'_{i}(1-\epsilon) = p_i
\end{aligned}
\]
\end{proof}

Fortunately, the concept of point mass aggregation has the appealing
property that the converse of Theorem~\ref{maintheorem1} holds as
well:

\begin{theorem}\label{thmscenarioeq}
Let $M$ be a scenario set. Then there exists a set $N$ of quadrant
requirements, such that any $P$ fulfilling the quadrant requirements
from $N$ can be written in the form $P={P'_M}^{pt}$ with a suitable
probability measure $P'$.
\end{theorem}

\begin{proof}
Denote the enhanced scenarios from $M$ by
$(d_1,p_1),(d_2,p_2),\ldots,(d_m,p_m)$. Define quadrants by $Q_j =
\{d_j\}$. Let $N$ be the set of quadrant requirements defined by
$Q_j$ and the associated probabilities $p_j$. We claim that $N$ has
the desired property. Indeed, if $P$ is a probability measure
fulfilling the quadrant requirements from $N$, in case $p_{M} < 1$
we can define
\[P' =\frac{1}{1-p_M} \left(P-\sum_{S\in M}{p_S \delta_{d_S}}\right)\,.\]
It is easily verified that $P'$ is a probability measure: The
positivity follows from the fact that $P$ fulfills the quadrant
requirements, and $P'(\mathbb{R}^n)=1$ due to the normalizing factor
$\frac{1}{1-p_M}$ in the definition of $P'$. Using
Definition~\ref{ptmassdef}, one can verify by a short calculation
that ${P'_M}^{pt}=P$.

If $p_{M} = 1$ then $P={P'_M}^{pt}$ for any probability measure
$P'$.
\end{proof}

One challenge for supervisors is to decide whether or not a scenario
can be excluded. An answer is given by

\begin{remark}
Theorem~\ref{pt_mass_agg_thm} together with
Theorem~\ref{thmscenarioeq} is very helpful if one is looking for
criteria to exclude scenarios for regulatory purposes while using
point mass aggregation: let $M$ be a scenario set. A subset $M'$ of
$M$ is called \emph{sufficient} for $P$ if ${P_{M'}}^{pt}$ satisfies
the quadrant requirements.
\end{remark}

As discussed in section \ref{Mathematicallyelegant}, the test
whether or not quadrant requirements are fulfilled, is free of
subjectivity and reproducible. Because of the equivalence stated in
Theorem~\ref{pt_mass_agg_thm} and \ref{thmscenarioeq}, these
properties carry over to scenarios when using point mass
aggregation.

\section{Further properties of scenario aggregation}

We consider the mapping on the set of probability
distributions defined by $P \mapsto {P_M}^{pt}$. One might ask whether,
knowing ${P_M}^{pt}$, we can reconstruct $P$.  The answer is provided in

\begin{theorem}
The mapping defined by $P\mapsto {P_M}^{pt}$ is injective for $p_{M}
< 1$. The image consists of all probability distributions satisfying
the quadrant requirement associated to $M$ by means of
Theorem~\ref{thmscenarioeq}.
\end{theorem}

\begin{proof}
The characterization of the image is a direct consequence of
Theorem~\ref{thmscenarioeq}. The injectivity can be seen by
observing that
\[
P =\frac{1}{1-p_M} \left({P_M}^{pt}-\sum_{S\in M}{p_S
\delta_{d_S}}\right),
\]
thus we have an explicit formula for $P$ given ${P_M}^{pt}$.
\end{proof}

\begin{remark}
Whilst the properties of the mapping induced by point mass
aggregation are easily studied, the properties of the corresponding
mapping induced by aggregation by shifting do not appear to so
clear.
\end{remark}

Next, we consider the properties of successive aggregation of
scenarios. Here we can see that even point-mass aggregation is not
so easily tractable.
Consider disjoint scenario sets $M_1, M_2$ respectively such that
$M_1\cup M_2$ is again a scenario set (i.e.\ total probability of
all scenarios $\leq 1$). Then one might ask whether
\[
P_{M_1\cup M_2}, (P_{M_1})_{M_2}, (P_{M_2})_{M_1}
\]
all lead to the same result; the same question may be asked for
point mass aggregation. The answer is that, even in the case of
point mass aggregation, the three terms may generally lead to different
results. This is even the case for two different enhanced scenarios, as can be seen from

\begin{example}
Take $P$ as $\delta_0$, the Dirac measure centered at $0$, and take $M_1, M_2$
as sets consisting of one enhanced scenario each, defined by $(d_1,p_1)$ and
$(d_2,p_2)$.

Then we have, for point mass aggregation
\[
\begin{aligned}
P_{M_1\cup M_2}^{pt} =(1-p_1-p_2)\delta_0+p_1\delta_{d_1}+p_2\delta_{d_2} \\
(P_{M_1}^{pt})_{M_2} = (1-p_1)(1-p_2)\delta_0+p_1(1-p_2)\delta_{d_1}+p_2\delta_{d_2} \\
(P_{M_2}^{pt})_{M_1} = (1-p_1)(1-p_2)\delta_0+p_1\delta_{d_1}+(1-p_1)p_2\delta_{d_2}
\end{aligned}
\]

All three terms generally lead to different results, so it is important when talking about
scenario aggregation to specify whether finitely many scenarios are to be aggregated
in one step or successively. Furthermore, if scenarios are aggregated successively, the
order needs to be specified. To avoid this, we are using the aggregation in one step
as a standard in this paper when aggregating a scenario set, which has the advantage
that we do not need to specify an order for our scenario sets.

Similar formulas are obtained when using aggregation by shifting.
\end{example}

So a supervisor should be very clear in what he is requiring, because
aggregation in one step may well lead to another result than
successive aggregation.

\section{Simulation based approaches and the quadrant definition}

During the proof of Theorem~\ref{thmscenarioeq}, we have made use of
quadrants consisting of one point only. Whilst theoretically useful,
such quadrant requirements lead to difficulties in the context of a
simulation based approach. The reason for this is that a simulation
generated based on continuous probability distributions will fulfill
such a quadrant requirement with probability $0$.

For this purpose, we introduce a somewhat restricted definition of
quadrants:

\begin{definition}
A quadrant $Q$ is called \emph{non-degenerate}, if it has non-zero
Lebesgue measure.
\end{definition}

Therefore, when practically working with quadrant requirements,
supervisors should use non-degenerate quadrants if they want to
enable companies to apply the usual simulation techniques.

\section{Scenario aggregation: a more general setting}
So far, we discussed two methods, aggregation by shifting and by
point mass. In this section we want to outline how these
results can be embedded in a more general framework. It turns out, that for
theoretical and practical reasons point mass aggregation excels by unique properties
within the class of aggregation methods based on mixing. We start with

\begin{definition}[$\phi$-aggregation]
Let $M$ be a scenario set, $P$ a probability measure on
$(\mathbb{R}^n,\mathfrak{B}_n)$, and $(\phi_{S})_{S}$ an arbitrary
family of measurable mappings $\mathbb{R}^{n} \rightarrow
\mathbb{R}^{n}$. Define
\[P^{\phi}_M = (1-p_M) P + \sum_{S\in M}{p_S {\phi_{S}}_{*} P}\]
We say that $P^{\phi}_M$ is obtained from $P$ by
\emph{$\phi$-aggregating} enhanced scenarios from scenario set $M$
to $P$.
\end{definition}

The two methods discussed above are easily included, as we can see
in

\begin{example}
Point mass aggregation is a $\phi$-aggregation for the family of
constant mappings $\phi_{S}: v \mapsto d_{S}$, and aggregation by
shifting is clearly induced by the family $\phi_{S}:=\tau_{d_S}$.
\end{example}

Example \ref{Counterexample} makes only use of \emph{one} property
of translations, their feature being \emph{distance preserving}
mappings. In order to generalize this counterexample, we need

\begin{definition}
A mapping $\phi: \mathbb{R}^{n} \rightarrow \mathbb{R}^{n}$ is
called \emph{expanding}, if $\vert x - y \vert \leq \vert \phi(x) -
\phi(y) \vert$ $\forall x, y \in \mathbb{R}^{n}.$
\end{definition}

To proceed, we need

\begin{lemma}\label{expanding}
An expanding mapping $\mathbb{R}^{n} \rightarrow \mathbb{R}^{n}$ is
always injective, and bijective if and only if it is continuous.
\end{lemma}

\begin{proof}
We first show injectivity. Let $x,y \in \mathbb{R}^n$, and assume
$\phi(x)=\phi(y)$. Then by definition of an expanding map, we have

\[
0 = \vert \phi(x) - \phi(y) \vert \geq \vert x-y \vert
\]

Thus we have $\vert x-y \vert =0$ and $x=y$.

Now we assume $\phi$ to be continuous, so we have to show
surjectivity. As $\mathbb{R}^n$ is connected, it suffices to show
that the image $\phi(\mathbb{R}^n)$ is both open and closed. First,
we note that the image is open due to the open mapping theorem. So
it remains to be shown that it is closed. For this, it is sufficient
to show that with any convergent sequence $(y_j)$, $y_j \in
\phi(\mathbb{R}^n)$, the limit $y$ is also contained in
$\phi(\mathbb{R}^n)$.

Choose $x_j \in \mathbb{R}^n$ such that $\phi(x_j)=y_j$. By the
definition of an expanding map, we have for any $j, k \in
\mathbb{N}$

\[
\vert x_j - x_k \vert \leq \vert \phi(x_j) - \phi(x_k) \vert = \vert
y_j - y_k \vert
\]

From this inequality and the fact that $(y_j)$ is a Cauchy sequency,
we conclude that $(x_j)$ is also a Cauchy sequence. Denote the limit
by $x$. Then by continuity, we conclude that $y=\phi(x)$, therefore
$y \in \phi(\mathbb{R}^n)$, and we have shown the surjectivity of
$\phi$.

Conversely, assume surjectivity. Then $\phi$ is bijective, so the
inverse mapping $\psi$ can be defined. It is obvious that the
definition of an expanding map implies that $\psi$ is a contraction.
Especially, $\psi$ is Lipschitz-continuous and therefore continuous.
Due to the open mapping theorem, $\psi$ maps open sets to open sets,
which means that the inverse image of an open set $U$under $\phi$ is
open, because $\phi^{-1}(U)= \psi (U)$. By definition, this means
that $\phi$ is continuous.
\end{proof}

There is a vast literature on expanding maps and on their
generalizations on Hilbert spaces other then $\mathbb{R}^{n}$, we
refer to (\cite{GR81}, \cite{SZ01}). Especially, the generalization
of Lemma~\ref{expanding} is only valid under additional assumptions
in the case of Hilbert or Banach spaces. The reason why the proof
given above cannot be carried over is that the open mapping theorem
does not hold in arbitrary Hilbert spaces.

With arguments similar as in example \ref{Counterexample} we can now
conclude

\begin{lemma}\label{Niceresult}
Let $M$ be a scenario set, and $\phi_S$ a family of measurable
expanding mappings, $p_{\max} \in [0,1]$ and $P$ such that $P(B) <
p_{\max}$ for all balls $B$ with fixed radius $R \geq 0$. Then
$P_M^{\phi}$ does not fulfill the quadrant requirement $(A,
p_{\max})$ for any quadrant $A$ contained in a ball $B$ with radius
$R / 2$. Furthermore, if $\phi$ is continuous, then $P$ never
fulfills the quadrant requirement $(A, p_{\max})$ for any quadrant
$A$ contained in a ball $B$ with radius $R$.
\end{lemma}

Thus, we should exclude the expanding mappings for the purpose of
scenario aggregation if we want to be sure that quadrant
requirements can be fulfilled. Note that any distance preserving
mapping is expanding by definition, so that translations and other
isometric mappings are naturally covered by Lemma~\ref{Niceresult}.

Point mass aggregation is induced by constant mappings, which are a
special case of contracting mappings $\phi$, where
\emph{contracting} is given by the dual notion of expanding, i.e.
$\vert x - y \vert \geq \vert \phi(x) - \phi(y) \vert$ $\forall x, y
\in \mathbb{R}^{n}$.

\begin{remark}
One might be tempted to try other contracting mapping $\phi_{S}$
than just the constant ones in order to satisfy quadrant
requirements by $\phi$-aggregation. In this case, however, we would
recall the arguments at the end of section \ref{casestudy}.
%and those at the beginning of section \ref{Translatingquadrant}.
Indeed, it turns out that point mass aggregation is the only one
induced by a contracting family without these caveats.
\end{remark}

So far, we only used mixing for aggregation of scenarios. One might
well ask the question, whether or not there are other algorithms of
aggregating them. Suppose there is an arbitrary input distribution
$P$, and an algorithm unknown to the supervisor (``black box'') which
processes $P$ in finite time into a distribution $P'$ which
satisfies an arbitrary finite set of quadrant requirements. This
``black box'' is equivalent to the supervisor's algorithm of point
mass aggregation which yields $P^{pt}$ in finite time: $P^{pt}$
might well be not equal to $P'$, but equivalent in the sense that
both $P'$ and $P^{pt}$ satisfy the given quadrant requirements. Thus
there is always an mixing algorithm which yields equivalent results
as the ``black box'' but is much more transparent then it.
%If we additionally allow countably many requirements on quadrants, then by lemma ? both algorithms even yield the identical result of $P'' = P'$.

All together, one is well advised not seeking for too long time for an alternative algorithm in order to fulfill
quadrant requirements, since it would be anyway equivalent to an
aggregation by mixing.

\section{Quadrant requirements and duality}

We have defined a quadrant $A$ as a finite intersection of affine
half-spaces in $\mathbb{R}^n$. Now we move from $\mathbb{R}^n$ to
$\mathfrak{M}:=\mathfrak{M}(\mathfrak{B}_{n},\mathbb{R})$ the set of
finite real-valued Borel-measures (not necessarily positive). For
$\mu_{1}, \mu_{2} \in \mathfrak{M}$ and $r \in \mathbb{R}$ we set
$(r \cdot \mu_{1})(A):= r \cdot \mu_{1}(A)$ and $(\mu_{1}+\mu_{2})(A):=
\mu_{1}(A) + \mu_{2}(A)$ for all $A \in \mathfrak{B}_{n}$, so
$\mathfrak{M}$ becomes a vector space over $\mathbb{R}$. The subset
of probability measures is denoted by
$\mathfrak{M}_{prob}:=\mathfrak{M}_{prob}(\mathfrak{B}_n,\mathbb{R})$.

We observe that the map evaluating at $A$
\[
\lambda_A: \mathfrak{M} \rightarrow \mathbb{R}, \mu \mapsto \mu(A)
\]
is a continuous linear form on $\mathfrak{M}$.

The notion of affine half-space and thus quadrants in definition
\ref{quadrants} carries easily over to arbitrary (not necessarily
finite dimensional) real vector spaces, and the notion of $\mu$
satisfying a quadrant requirement can easily be generalized from
probability measures to any real-valued measure: $\mu$ satisfies the
quadrant requirement $(A,p)$ if $\mu(A) \geq p$.

\begin{definition}
We say that $\mu$ belongs to the \emph{acceptance set} of quadrant
requirement $(A,p)$ if $\mu$ satisfies $(A,p)$; we use the same
terminology in case of an arbitrary set of quadrant requirements.
\end{definition}

The condition for belonging to the acceptance set of a single
quadrant requirement can be rewritten as follows
\[
\lambda_A(\mu) \geq p
\]
From this, we immediately obtain the following

\begin{lemma}\label{acceptancelemma}
The acceptance set of a single quadrant requirement is an affine
half-space in $\mathfrak{M}$.
\end{lemma}

Therefore, while a quadrant is a finite intersection of affine
half-spaces in $\mathbb{R}^n$, the acceptance sets of a finite set
of quadrant requirements is an intersection of finitely many affine
half-spaces in $\mathfrak{M}$.

\begin{theorem}\label{acceptancethm}
The acceptance set of an arbitrary set of quadrant requirements is
closed and convex, both in $\mathfrak{M}_{prob}$ and in
$\mathfrak{M}$.
\end{theorem}

\begin{proof}
In the case of $\mathfrak{M}$, the assertion follows from the fact
that due to Lemma~\ref{acceptancelemma} any acceptance set is an
intersection of affine half-spaces. To extend it to
$\mathfrak{M}_{prob}$, we note that the condition for a real-valued
measure to be a probability measure can also be written as follows
as an intersection of half-spaces:
\[
\forall A \in \mathfrak{B}_n: \lambda_A(\mu)\geq 0,
\lambda_{\mathbb{R}^n}(\mu)=1
\]
Hence also in this case, the acceptance sets are intersections of
affine half-spaces and therefore convex.
\end{proof}

\begin{remark}
Theorem~\ref{acceptancethm} means that a supervisor may only use
quadrant requirements to describe its regulatory requirements if the
set of distributions which is accepted by the supervisor is closed
and convex.
\end{remark}

Conversely, one may ask the question whether any convex set of
measures can be written as acceptance set of a set of quadrant
requirements. For this purpose, we need to study the dual space of
$\mathfrak{M}$, and it turns out that it is too rich in the sense
that not all closed convex sets may be described as quadrant
requirements in the sense of our definition. However, we can
generalize the notion of quadrant requirement as follows to obtain a
result.

\begin{definition}
Let $\mu \in \mathfrak{M}$, $g: \mathbb{R}^n \rightarrow \mathbb{R}$
a function such that its $\mu$-integral exists, and $p \in
\mathbb{R}$. A \emph{generalized requirement} defined by $(g,p)$ is
a condition of the form
\[
\int g d\mu \geq p
\]
and $\mu$ is said to belong to the acceptance set of $(g,p)$; we use
similar terminology in the case of a set of generalized
requirements.
\end{definition}

\begin{remark}
One can easily see that quadrant requirements are a special case, by
taking $g=\chi_A$ to be the characteristic function of a quadrant
$A$.
\end{remark}

\begin{theorem}
Any closed convex subset of $\mathfrak{M}$ can be described as
acceptance set of a set of generalized requirements.
\end{theorem}

\begin{proof}
It is a known fact that the dual space of $\mathfrak{M}$ can be
identified with $L^\infty(\mathbb{R}^n)$ (with respect to the
Lebesgue measure), hence any linear form on $\mathfrak{M}$ can be
expressed by integrating over a function $g \in
L^\infty(\mathbb{R}^n)$ (cf. \cite{LA93}, Chapter VII, Theorem~2.2
in conjunction with Corollary~4.3).

Similarly, by the Hahn-Banach separation theorem (cf. \cite{LA93},
Appendix to Chapter IV, Theorem 1.2), any convex subset of a Banach
space can be expressed as an intersection of affine half-spaces.

Taking these two facts together, the theorem follows immediately.
\end{proof}

\begin{remark}
We see that a supervisor may have to use generalized requirements in
case the acceptance set is more complicated in the sense that
inspection of weights on (possibly infinitely many) quadrants is not
sufficient to decide whether a distribution is accepted or not. In
this case, as a first step between quadrant requirements and
arbitrary generalized requirements, it might be useful to try
whether it is possible to get along with considering generalized
requirements defined by step functions $g$ instead of characteristic
functions of quadrants.
\end{remark}

\section{Conclusion}

We introduced and discussed scenarios and different methods for aggregating them at risk factor level. We defined quadrant requirements and investigated their relationship and compatibility with different scenario aggregation methods. It turns out that aggregation methods based on contractive mappings, and especially point mass aggregation, play an important role in this respect. Furthermore, we studied generalized requirements and showed that they may be used to describe any closed convex acceptance set of risk factor distributions.

\begin{appendix}
\section{Appendix: Proof of Lemma~\ref{convex sets}}

Let $S\subseteq \mathbb{R}^n$. We define the asymptotic cone of $S$, denoted by
$\mathrm{Cone}(S)$. It consists of all $x\in \mathbb{R}^n$ such that $y+tx\in
S$ for all $t>0$ and $y \in S$. Some more background
of the asymptotic cone together with some useful properties can be found in \cite{HL01}, p. 39.

Initially, we note that, for a fixed dimension $n$, it suffices to show the lemma for \emph{closed} convex sets.
Indeed, assume the lemma is true in this case, and let $S$ be an arbitrary convex set.

$\bar{S}$ is closed and convex, hence by assumption, we can apply the lemma to $\bar{S}$. Therefore, $\bar{S}$ is either two-sided constrained, or it contains arbitrarily large balls.
In case $\bar{S}$ is two-sided constrained, then so is $S$, so we are done. Otherwise, $\bar{S}$ contains balls $B_{2R}$ of radius $2R$ for any $R>0$.
We claim that the corresponding smaller ball $B_R$ of radius $R$ around the same point is then contained in $S$. Indeed, we can find finitely many points $z_j\in B_{2R}$
such that the convex ball of radius $2R$ is contained in the convex hull of these points. for each point, we can find $z_j'\in S$ such that $\vert z_j - z_j' \vert < \epsilon$.
By choosing $\epsilon>0$ sufficiently small, $B_R$ is also contained in the convex hull of the points $z_j'$, which itself is a subset of $S$ by convexity.
Hence $S$ contains a ball of radius $R$.

Next, we proceed to show the lemma in dimension $n$, and from what we have seen above, we may assume that $S$ is closed. We proceed by induction on $n$, the start of the induction is obvious and has been
described in Example~\ref{linearformsonR}. So we can assume that the lemma is true in all dimensions $<n$, and for arbitrary convex sets (not necessarily closed).

We choose a maximal linear independent subset of $\mathrm{Cone}(S)$, denoted by $x_1, \ldots x_k$.

If $k=n$, we are finished, because $S$ then contains arbitrarily large, non-degenerate n-simplices, and thus arbitrarily large balls.

Assume $k<n$. Let $L:\mathbb{R}^n\rightarrow \mathbb{R}^{n-k}$ be a linear map such that the kernel is the subspace generated by $v_1,\ldots,v_k$.
Define $S'=L(S)$. $S'$ is also a convex set.

If $S'$ is two-sided constrained, then so is $S$, because if $\lambda$ is a non-zero linear form on $\mathbb{R}^{n-k}$ such that $\lambda(S')$ is bounded,
then $\lambda':=\lambda\circ L$ is a non-zero linear form on $\mathbb{R}^n$ such that $\lambda'(S)$ is bounded.

Hence we can assume $S'$ is not two-sided constrained.
As $S'$ lies in a vector space of dimension $<n$, we can apply the lemma to $S'$ by induction assumption, concluding that $S'$ contains arbitrarily large balls, and thus an arbitrarily large $n-k$ cube spanned
by points $y_1,\ldots,y_{l}\in S'$. Assume $x_j\in S$ is mapped to $y_j$ under $L$. Then by the invariance of the asymptotic cone (cf. Proposition 2.2.1 in \cite{HL01};
note that we are now making use of the fact that $S$ is closed) we see that
$S$ contains a one-sided cylinder over the cube spanned by the points $x_1,\ldots,x_{l}$ and the directions $v_1,\ldots,v_k$. As we can make the cube
arbitrarily large, it is easily verified that this cylinder, and thus $S$, contains arbitrarily large $n$-cubes, and thus arbitrarily large balls.

For elaborating the last step, we have make use of the fact that by the equivalence of norms on $\mathbb{R}^n$, the property of containing large balls (or cubes) is preserved under any
linear isomorphism, and we may thus without loss of generality assume that the $v_k$ are the standard unit vectors, in which case the assertion follows immediately.

$\bar{S}$ is closed and convex, hence from what we have just proved, $\bar{S}$ is either two-sided constrained, or it contains arbitrarily large balls.
In case $\bar{S}$ is two-sided constrained, then so is $S$, so we are done. Otherwise, $\bar{S}$ contains balls $B_{2R}$ of radius $2R$ for any $R>0$.
We claim that the corresponding smaller ball $B_R$ of radius $R$ around the same point is then contained in $S$. Indeed, we can find finitely many points $z_j\in B_{2R}$
such that the convex ball of radius $2R$ is contained in the convex hull of these points. for each point, we can find $z_j'\in S$ such that $\vert z_j - z_j' \vert < \epsilon$.
By choosing $\epsilon>0$ sufficiently small, $B_R$ is also contained in the convex hull of the points $z_j'$, which itself is a subset of $S$ by convexity.
Hence $S$ contains a ball of radius $R$.

\end{appendix}

\section*{Acknowledgements}
   The authors wish to thank T. Cooke, D. Finnis, C. Lang and M. Schmutz for their useful feedback and very helpful hints.

====\\
Andreas Haier \\
Milit\"arstrasse 44 \\
CH-3014 Bern \\
Switzerland \\
\\
Thorsten Pfeiffer \\
Fellenbergstrasse 17 \\
CH-3012 Bern\\
Switzerland

\end{document}